\documentclass[submission,copyright,creativecommons]{eptcs}
\providecommand{\event}{SOS 2007} 

\usepackage{iftex}
\usepackage{amssymb}
\usepackage{amsthm}
\usepackage{amsopn}
\usepackage{amsmath}
\usepackage{bussproofs}
\ifpdf
  \usepackage{underscore}         
  \usepackage[T1]{fontenc}        
\else
  \usepackage{breakurl}           
\fi

\title{Semantic Analysis of Subexponential Modalities in Distributive Non-commutative Linear Logic}
\author{Daniel Rogozin
\institute{University College London \\ London, UK}
\email{d.rogozin@ucl.ac.uk}
}

\newcommand{\titlerunning}{Subexponential Modalities}
\newcommand{\authorrunning}{D. Rogozin}

\hypersetup{
  bookmarksnumbered,
  pdftitle    = {\titlerunning},
  pdfauthor   = {\authorrunning},
  pdfsubject  = {EPTCS},               
}

\newtheorem{theorem}{Theorem}
\newtheorem{lemma}[theorem]{Lemma}
\newtheorem{proposition}[theorem]{Proposition}%
\theoremstyle{definition}
\newtheorem{definition}[theorem]{Definition}%
\newtheorem{remark}[theorem]{Remark}

\providecommand{\keywords}[1]
{
  \textbf{\textit{Keywords---}} #1
}

\begin{document}
\maketitle

\begin{abstract}
In this paper, we consider the full Lambek calculus enriched with subexponential modalities in a distributive setting. We show that the distributive Lambek calculus with subexponentials is complete with respect to its Kripke frames via canonical extensions.
In this approach, we consider subexponentials as ${\bf S4}$-like modalities and each modality is interpreted with a reflexive and transitive relation similarly to usual Kripke semantics.

\end{abstract}

\section{Introduction}

Substructual logics are a kind of non-classical logic lacking some of the structual rules: weakening, contraction and exchange, we refer the reader to \cite{restall2002introduction} to have a more systematic introduction. The Lambek calculus is a logic with none of these rules, which was introduced initially by Lambek for modelling natural language grammar by proof-theoretic means \cite{lambek1958mathematics}. From an algebraic point of view, the Lambek calculus is the logic of residuated semigroups, which are also connected to other substructural logics such as relevant or linear logics, see e.g. \cite{ono2003substructural,urquhart1996duality}.

The $!$-modality originally comes from linear logic, where $!$ introduces lacking structural rules in a restricted way, see \cite{girard1987linear, lafont1987linear}. Such modal enrichments were also discussed in the context of resource management in computation based on linear types, see \cite{lafont1988linear, wadler1993taste}.
Enrichments of the Lambek calculus with the $!$-modality can be also motivated in terms of mathematical lingustics, see \cite{kanovich2021multiplicative}.

The polymodal expansion of the full Lambek calculus (that is, the Lambek calculus with additive connectives and constants) with subexponentials has been introduced by Kanovich, Kuznetsov, Scedrov and Nigam \cite{kanovich2019subexponentials} to generalise linear logical frameworks used for linear authorisation logics and concurrent programming languages \cite{nigam2009algorithmic}.
At the moment, subexponential modalities in non-commutative linear logics have not been analysed comperehensively in terms of semantics and the current known results are rather about proof-theoretic and computability aspects, see  \cite{kanovich2019subexponentials}.

In this paper, we show the completeness result for subexponential expansions of distributive non-commutative linear logic.
As far as subexponential $!$-modalites obey some formal properties of the $\Box$-modality from classical modal logic,
we consider them as Kripke-style necessity operators as in usual Kripke semantics of modal logic:
\begin{center}
$\mathcal{M}, x \models \: ! \varphi$ iff $\forall y \: (x R y \Rightarrow \mathcal{M}, y \models \varphi)$
\end{center}
Generally, the full Lambek calculus is non-distributive, so the law of the form
\begin{center}
$\varphi \land (\psi \lor \theta) \vdash (\varphi \land \psi) \lor (\varphi \land \theta)$
\end{center}
is not valid unless we require this principle as an extra-axiom. So we consider the distributive full Lambek calculus enriched with subexponentials in terms of Kripke semantics with the standard truth definition for disjunction and conjunction.
Notice that the truth definition for conjuction along with the Kripke-style definition for $!$ semantically imply that the following principals should be valid as well, so we take them as the additional axioms:
\begin{itemize}
\item $\top \: \vdash \: ! \top$,
\item $! \varphi \: \land \: ! \psi \vdash \: !(\varphi \land \psi)$.
\end{itemize}

We prove that the distributive full Lambek calculus with subexponential modalities is Kripke complete by showing that the corresponding variety of distributive residuated lattices with modal operators is closed under canonical extensions and conclude Kripke-completeness from canonicity.
In turn, canonical extensions are a kind of completions in algebraic logic initially introduced for the purposes of extending Stone representation theorem for Boolean algebras with operators \cite{jonsson1951boolean}, but canonical extensions were further generalised for
bounded (non)distributive lattices with operators \cite{gehrke2004bounded, gehrke2011view, conradie2019algorithmic}.

\section{The full Lambek calculus and subexponentials}
In this section, we recall the reader some preliminary notions related to the full Lambek calculus, see \cite[Chapter 2]{galatos2007residuated} for details.

Assume we have a fixed countable set of propositional variables $\operatorname{PV} = \{ p_i \: | \: i < \omega \}$. The set of formulas is generated by the following grammar:
\begin{center}
$\varphi ::= \bot \: | \: \top \: | \: {\bf 1} \: | \: p \: | \: (\varphi \bullet \varphi) \: | \: (\varphi \setminus \varphi) \: | \: (\varphi / \varphi) \: | \: (\varphi \lor \varphi) \: | \: (\varphi \land \varphi)$
\end{center}
A family of substructural modal logics we are going to consider further enrich the full Lambek calculus, which is defined the following way.

\begin{definition} \emph{The full Lambek calculus} is defined as the smallest set of pairs of formulas $\varphi \vdash \psi$ that contains the following axioms and is closed under the inference rules:

\begin{itemize}
\item $p \vdash \top$
\item $\bot \vdash p$
\item $p \bullet (q \bullet r) \dashv \vdash (p \bullet q) \bullet r$
\item $p \dashv \vdash {\bf 1} \bullet p \dashv \vdash p \bullet {\bf 1}$
\item $p_i \vdash p_1 \lor p_2$, for $i = 1, 2$
\item $p_1 \land p_2 \vdash p_i$, for $i = 1, 2$
\end{itemize}

\begin{center}
\end{center}
\begin{minipage}{0.5\textwidth}
\begin{flushleft}

\begin{prooftree}
\AxiomC{$\varphi \vdash \psi$}
\AxiomC{$\varphi \vdash \theta$}
\BinaryInfC{$\varphi \vdash \psi \land \theta$}
\end{prooftree}

\begin{prooftree}
\AxiomC{$\varphi \vdash \psi$}
\AxiomC{$\theta \vdash \tau$}
\BinaryInfC{$\varphi \bullet \theta \vdash \psi \bullet \tau$}
\end{prooftree}

\begin{prooftree}
\AxiomC{$\varphi \bullet \psi \vdash \theta$}
\doubleLine
\UnaryInfC{$\psi \vdash \varphi \setminus \theta$}
\end{prooftree}

\end{flushleft}
\end{minipage}\hfill
\begin{minipage}{0.5\textwidth}
\begin{flushright}

\begin{prooftree}
\AxiomC{$\varphi \vdash \theta$}
\AxiomC{$\psi \vdash \theta$}
\BinaryInfC{$\varphi \lor \psi \vdash \theta$}
\end{prooftree}

\begin{prooftree}
\AxiomC{$\varphi(p) \vdash \psi(p)$}
\UnaryInfC{$\varphi[p := \theta] \vdash \psi[p := \theta]$}
\end{prooftree}

\begin{prooftree}
\AxiomC{$\varphi \bullet \psi \vdash \theta$}
\doubleLine
\UnaryInfC{$\psi \vdash \theta / \varphi$}
\end{prooftree}

\end{flushright}
\end{minipage}
\begin{center}
\begin{prooftree}
\AxiomC{$\varphi \vdash \psi$}
\AxiomC{$\psi \vdash \varphi$}
\BinaryInfC{$\varphi \vdash \theta$}
\end{prooftree}
\end{center}

\end{definition}

The \emph{distributive} full Lambek calculus extends the full Lambek calculus with the following extra axiom
\begin{center}
$p \land (q \lor r) \vdash (p \land q) \lor (p \land r)$.
\end{center}

Algebraically, the full Lambek calculus is the logic of residuated lattices. \emph{A residuated lattice} is an algebra $\mathcal{L} = (L, \wedge, \vee, \cdot, \cdot, \setminus, /, \top, \bot, \epsilon)$ such that $(L, \wedge, \vee, \top, \bot)$ is a bounded lattice, $(L, \cdot, \epsilon)$ is a monoid and $\setminus$ and $/$ are binary operations (\emph{residuals}) such that for all $a, b, c \in L$:
\begin{center}
$b \leq a \setminus b \leftrightarrow a \cdot b \leq c \leftrightarrow a \leq c / b$
\end{center}
A residuated lattice is \emph{distributive} if its lattice reduct is distributive.

Note that the classes of distributive and non-distributive residuated lattices are known to be varieties \cite[Lemma 2.3]{jipsen2002survey}, so both of these classes are equationally axiomatisable. Thus the (distributive) full Lambek calculus is complete with respect to the variety of (distributive) residuated lattices. Standarly, each pair of formulas $\varphi \vdash \psi$ corresponds to an algebraic inequality $\varphi \leq \psi$, which is, in turn, an algebraic identity of the form $\varphi \wedge \psi = \psi$.

To enrich the full Lambek calculus with a family of modalities $(!_i)_{i \in \Sigma}$ for some $\Sigma \neq \emptyset$, we extend the grammar of formulas as follows:
\begin{center}
$!_i \varphi$ is a formula for each $i \in \Sigma$ whenever $\varphi$ is a formula.
\end{center}

Let $(I, \preceq)$ be a preorder. A \emph{subexponential signature} is a tuple $\Sigma = (I, \preceq, W, E, C)$, where $W, E, C$ are upward closed subsets of $I$ such that $W \cap C \subseteq E$. Each of these distinguished subsets corresponds to those modalities that emulate a particular structural rule. That is, subexponentials introducing the weakining rule are indexed by elements from $W$ and the same for the exchange and contraction rules. The requirement $W \cap C \subseteq E$ means that if a subexponential allows weakining and exchanging, then it also allows contracting, see \cite{kanovich2019subexponentials}.

Originally, the Lambek calculus with subexponentials was introduced in \cite{kanovich2019subexponentials} as follows.
\begin{definition}
\emph{The full Lambek calculus with subexponential modalities} over a signature $\Sigma$, denoted as $\operatorname{SMALC_{\Sigma}}$, is the defined by the following axioms and inference rules, for all $i,j, k \in \Sigma$:
\begin{enumerate}
\item The axioms and inference rules of the full Lambek calculus,
\item $!_i p \: \vdash \: p$ and $!_i p \: \vdash \: !_i !_i p$,
\item $!_i p \: \bullet \: !_j q \vdash \: !_k (p \bullet q)$ such that $k \succeq i, j$,
\item ${\bf 1} \vdash \: !_{i} {\bf 1}$,
\item $!_i p \: \bullet q \vdash \: !_i p \bullet q \: \bullet \: !_i p$ and $q \: \bullet \: !_i p \vdash \: !_i p \bullet q \: \bullet \: !_i p$, for $i \in C$,
\item $!_i p \vdash {\bf 1}$, for $i \in W$,
\item $!_i p \bullet \: q \dashv \vdash \: q \: \bullet \: !_i p$ for $i \in E$,
\item From $\varphi \vdash \psi$ infer $!_{i} \varphi \: \vdash \: !_{j} \psi$ whenever $j \preceq i$.
\end{enumerate}
\end{definition}

\begin{remark}
  We also can drop the ${\bf S4}$-axioms and all results of this paper will be preserved in that case.
  However, these axioms are important as the promotion and dereliction principles from linear logic and modal enrichments of the Lambek calculus.
\end{remark}

\begin{remark}
In \cite{kanovich2019subexponentials}, $\operatorname{SMALC_{\Sigma}}$ was introduced as a Gentzen-style sequent calculus, but our version is more Hilbert-like. One can show that both formalisms are equivalent standardly.
\end{remark}

\begin{remark}
Instead of the traditional contraction axiom of the form $!_i \: p \vdash \: !_i p \: \bullet \: !_i p$, we use its more general version introduced in \cite{kanovich2019subexponentials}, which is called \emph{non-local contraction}.
\end{remark}

As we have already discussed in the first section, we extend $\operatorname{SMALC_{\Sigma}}$ with additional axioms to define the distributive version of the full Lambek calculus with subexponentials.
\begin{definition}
\emph{The distributive full Lambek calculus with subexponential modalities} over a subexponential signature $\Sigma$, denoted as $\operatorname{DSMALC_{\Sigma}}$, extends $\operatorname{SMALC_{\Sigma}}$ with the following extra axioms, for each $i \in \Sigma$:
\begin{enumerate}
\item $p \land (q \lor r) \vdash (p \land q) \lor (p \land r)$,
\item $!_i p \: \land \: !_i q \vdash \: !_i (p \land q)$,
\item $\top \vdash \: !_i \top$.
\end{enumerate}
\end{definition}

For algebraic semantics, we define $\Sigma$-algebras for a subexponential signature $\Sigma$ by rewriting the axioms of $\operatorname{DSMALC_{\Sigma}}$ as algebraic inequalities.

\begin{definition}
Let $\Sigma$ be a subexponential signature, a $\Sigma$-algebra is an algebra $\mathcal{R} = (R, \wedge, \vee, \cdot, \setminus, /, (!_i)_{i \in \Sigma},\\ \top, \bot, \epsilon)$ such that, for each $i, j, k \in \Sigma$:
\begin{enumerate}
\item $(R, \wedge, \vee, \cdot, \setminus, /, \top, \bot, \epsilon)$ is a distributive residuated lattice,
\item $!_i$ preserves finite infima, $\epsilon \leq \: !_i \epsilon$ and $!_i !_i a = \: !_i a \leq a$ for all $a \in \mathcal{R}$,
\item If $k \preceq i, j$, then $!_i a \: \cdot \: !_j b \leq \: !_k (a \cdot b)$ for all $a, b \in \mathcal{R}$,
\item $!_i a \: \cdot b \leq \: !_i a \cdot b \: \cdot \: !_i a$ and $b \: \cdot \: !_i a \leq \: !_i a \cdot b \: \cdot \: !_i a$, for $a, b \in \mathcal{R}$ and $i \in C$,
\item $!_i a \leq \epsilon$, for $a \in \mathcal{R}$ and $i \in W$,
\item $!_i a \cdot \: b = \: b \: \cdot \: !_i a$ for $a, b \in \mathcal{R}$ and $i \in E$.
\end{enumerate}
\end{definition}

Let $\mathcal{R}$ be a $\Sigma$-algebra, an algebraic interpretation is a function $[\![.]\!] : \operatorname{PV} \to \mathcal{R}$ that commutes with connectives usually. Clearly that the class of all $\Sigma$-algebras is a variety and $\operatorname{DSMALC_{\Sigma}}$ is complete with respect to that class. We denote the free algebra with $\omega$ generators from the variety of $\Sigma$-algebras as $F_{\Sigma}$. The fact that $\varphi \vdash \psi$ is provable in $\operatorname{DSMALC_{\Sigma}}$ iff $F_{\Sigma} \models [\![\varphi]\!] \leq [\![\psi]\!]$ can be thought as folklore from universal algebra.

\section{Kripke semantics}

In this subsection, we introduce relational semantics for $\operatorname{DSMALC_{\Sigma}}$ in the fashion of Routley-Meyer models for relevant logic (see \cite{routley1972semantics, seki2003sahlqvist}) and well as for other substructural logics (see \cite{allwein1993kripke}).

First of all, we define ternary frames for the distributive full Lambek calculus.

\begin{definition}
A ternary Kripke frame is a structure $F = (W, \leq, R, \mathcal{O})$ where $R \subseteq W^3$, $\mathcal{O} \subseteq W$ and $(W, \leq)$ is a poset such that for all $u, v, w, u', v', w'$:
\begin{itemize}
\item $\exists x \in W \: R u w x \: \& \: R x u' v' \Leftrightarrow \exists y \in W \: R w u' y \: \& \: R u y v'$,
\item $Ruvw \: \& \: u' \leq u \Rightarrow Ru'vw$,
\item $Ruvw \: \& \: v' \leq v \Rightarrow Ruv'w$
\item $Ruvw \: \& \:w \leq w' \Rightarrow Ruvw'$,
\item $O$ is an upward closed subset such that $\forall o \in O \: R v o w \Leftrightarrow R o v w$ and $v \leq w \Rightarrow \exists o \in O \: R v o w$,
\end{itemize}
\end{definition}

One can associate a distributive residuated lattice with every ternary Kripke frame. Given a ternary frame $\mathcal{F} = (W, \leq, R, \mathcal{O})$, its dual \emph{complex algebra} is the algebra $\operatorname{Cm}(\mathcal{F}) = (\operatorname{Up}(W), \cap, \cup, \cdot, \setminus, /, \mathcal{O}, \emptyset, W)$ where:
\begin{enumerate}
\item $\operatorname{Up}(\mathcal{F})$ is the set of all upward closed subsets of $\operatorname{Up}(W)$,
\item $A \cdot B = \{ w \in W \: | \: \exists u, v \in W \: R u v w \: \& \: u \in A \: \& \: v \in B\}$,
\item $A \setminus B = \{ w \in W \: | \: \forall u, v \in W \: R u w v \: \& \: u \in A \Rightarrow v \in B \}$,
\item $A / B = \{ w \in W \: | \: \forall u, v \in W \: R w u v \: v \in B \Rightarrow u \in A \}$.
\end{enumerate}
It is readily checked that $\operatorname{Cm}(\mathcal{F})$ is well-defined since if $A, B$ are upward closed, so are $A \cdot B$, $A \setminus B$ and $A / B$. Also, using the definition of a ternary frame, one can show that $\operatorname{Cm}(\mathcal{F})$ is indeed a distributive residuated lattice.
\begin{definition}
Let $\mathcal{F}$ be a ternary frame and let $\Sigma$ be a subexponential signature. A \emph{$\Sigma$-frame} is an expansion of $\mathcal{F}$ with a family of binary relations $(R_i)_{i \in \Sigma}$ such that for all $i, j, k \in \Sigma$:
\begin{enumerate}
\item For all $i \in \Sigma$, $(\mathcal{F}, R_i)$ is a preorder,
\item For all $u, v, w, w' \in \mathcal{F}$, if $k \preceq i, j, Ruvw \: \& \: w R_{k} w'$, then $\exists x, y \in W \: R x y w' \: \& \: u R_{i} x \: \& \: v R_{j} y$,
\item $i \preceq j$ implies $R_j \subseteq R_i$,
\item $\mathcal{O} \subseteq [R_i] \mathcal{O}$.
\item For $A, B \in \operatorname{Up}(\mathcal{F})$ one has $[R_{i}] A \cdot B = B \cdot [R_{i}] A$ whenever $i \in E$,
\item For $A, B \in \operatorname{Up}(\mathcal{F})$ one has $[R_{i}] A \cdot B \subseteq [R_{i}] A \cdot B \cdot [R_{i}] A$ and $B \cdot [R_{i}] A \subseteq [R_{i}] A \cdot B \cdot [R_{i}] A$ whenever $i \in C$,
\item For $A \subseteq \operatorname{Up}(\mathcal{F})$ one has $[R_{i}] A \subseteq \mathcal{O}$ whenever $i \in W$.
\end{enumerate}
where $[R_{i}]A = \{ u \in W \: | \: \forall w \in W \: (u R_{i} w \Rightarrow u \in A) \}$ for $i \in \Sigma$.
\end{definition}
Given a $\Sigma$-frame $\mathcal{F} = (W, \leq, (R_{i})_{i \in \Sigma}, R, \mathcal{O})$, the \emph{complex algebra} of $\mathcal{F}$ is the algebra \\ $\operatorname{Cm}(\mathcal{F}) = (\operatorname{Up}(W), \cap, \cup, \cdot, \setminus, /, ([R_i])_{i \in \Sigma}, \mathcal{O}, \emptyset, W)$. Frames and complex algebras are connected with each other as usual in duality theory:

\begin{proposition}~\label{complex}
Let $\mathcal{F}$ be a $\Sigma$-frame, then $\operatorname{Cm}(\mathcal{F})$ is a $\Sigma$-algebra.
\end{proposition}

\begin{proof}
Follows from the definition of a $\Sigma$-frame.
\end{proof}

\begin{definition} \emph{A Kriple model} is a tuple $\mathcal{M} = (\mathcal{F}, \vartheta)$ where $\mathcal{F}$ is a modal ternary Kripke frame and $\vartheta : \operatorname{PV} \to \operatorname{Up}(W, \leq)$ is a valuation map. The truth definition is inductive:

\begin{itemize}
\item $\mathcal{M}, w \models p$ iff $w \in \vartheta(p)$,
\item $\mathcal{M}, w \not\models \bot$,
\item $\mathcal{M}, w \models \top$,
\item $\mathcal{M}, w \models {\bf 1}$ iff $w \in O$,
\item $\mathcal{M}, w \models \: !_i \varphi$ iff $\forall u \in R_{i}(w) \: \mathcal{M}, u \models \varphi$,
\item $\mathcal{M}, w \models \varphi \lor \psi$ iff $\mathcal{M}, w \models \varphi$ or $\mathcal{M}, w \models \psi$,
\item $\mathcal{M}, w \models \varphi \land \psi$ iff $\mathcal{M}, w \models \varphi$ and $\mathcal{M}, w \models \psi$,
\item $\mathcal{M}, w \models \varphi \bullet \psi$ iff $\exists u, v \in W \: R u v w \: \& \: \mathcal{M}, u \models \varphi \: \& \: \mathcal{M}, v \models \psi$,
\item $\mathcal{M}, w \models \varphi \setminus \psi$ iff $\forall u, v \in W \: R u w v \: \& \: \mathcal{M}, u \models \varphi \Rightarrow \mathcal{M}, v \models \psi$,
\item $\mathcal{M}, w \models \psi / \varphi$ iff $\forall u, v \in W \: R w u v \: \& \: \mathcal{M}, u \models \varphi \Rightarrow \mathcal{M}, v \models \psi$,
\item $\mathcal{M}, w \models \varphi \vdash \psi$ iff $\mathcal{M}, w \models \varphi$ implies $\mathcal{M}, w \models \psi$.
\end{itemize}

As usual, $\mathcal{M} \models \varphi \vdash \psi$ iff $\mathcal{M}, w \models \varphi \vdash \psi$ for each $w \in \mathcal{M}$.
\end{definition}

Let $\mathcal{F}$ be a $\Sigma$-frame, then $\mathcal{F} \models \varphi \vdash \psi$ iff $(\mathcal{F}, \vartheta) \models \varphi \vdash \psi$ for each $\vartheta$. Let $T$ be a set of sequents, then $\mathcal{F} \models T$ iff $\mathcal{F} \models \varphi \vdash \psi$ for each $\varphi \vdash \psi \in T$.

Given a model $\mathcal{M} = (W, \leq, R, (R_i)_{i \in \Sigma}, \mathcal{O}, \vartheta)$, define its truth set $[\![\varphi]\!] = \{ w \in W \: | \: \mathcal{M}, w \models \varphi \}$. One can show that the set of truth sets form a subalgebra of the complex algebra $\operatorname{Cm}(\mathcal{F})$, where $\mathcal{F}$ is the underlying $\Sigma$-frame of a model $\mathcal{M}$. To be more precise, the following fact standardly holds:

 \begin{proposition} \label{val}
 Let $\mathcal{M} = (W, \leq, R, (R_i)_{i \in I}, \mathcal{O}, \vartheta)$ be a model, then:

 \begin{enumerate}
 \item $[\![p]\!] = \vartheta(p)$,
 \item $[\![\bot]\!] = \emptyset$,
 \item $[\![\top]\!] = W$,
 \item $[\![{\bf 1}]\!] = \mathcal{O}$,
 \item $[\![\varphi \lor \psi]\!] = [\![\varphi]\!] \cup [\![\psi]\!]$
 \item $[\![\varphi \land \psi]\!] = [\![\varphi]\!] \cap [\![\psi]\!]$,
 \item $[\![\varphi \bullet \psi]\!] = [\![\varphi]\!] \cdot [\![\psi]\!]$,
 \item $[\![\varphi \setminus \psi]\!] = [\![\varphi]\!] \setminus [\![\psi]\!]$
 \item $[\![\varphi / \psi]\!] = [\![\varphi]\!] / [\![\psi]\!]$,
 \item $[\![!_i \varphi]\!] = [R_i][\![\varphi]\!]$ for each $i \in \Sigma$,
 \item $\mathcal{M} \models \varphi \vdash \psi$ iff $[\![\varphi]\!] \subseteq [\![\psi]\!]$.
 \end{enumerate}
 \end{proposition}

\begin{theorem} (Soundness)

Let $\mathcal{F}$ be a $\Sigma$-frame and $\vartheta$ a valuation, then $\mathcal{F} \models \operatorname{DSMALC}_{\Sigma}$.
\end{theorem}

\begin{proof}
We check only subexponential axioms, for the rest of the axioms and inference rules, the proof is similar to \cite[Theorem 2]{seki2003sahlqvist}.

Let $w \in \mathcal{F}$ and $i, j, k \in \Sigma$.
\begin{enumerate}

\item We show that $\mathcal{M}, w \models \: !_{i} p \: \bullet \: !_{j} q \vdash \: !_{k} (p \bullet q) \in L$.

Assume that $\mathcal{M}, w \models \: !_{i} p \: \bullet \: !_{j} q$. Then there are $u, v \in \mathcal{F}$ such that $R u v w$, $\mathcal{M}, u \models \: !_{i} p$ and $\mathcal{M}, v \models \: !_{j} q$. Take any $w' \in R_{k}(w)$. Then $R u v w$ and $w R_{k} w'$ imply that there are $x, y \in \mathcal{F}$ such that $R x y w'$, $u R_{i} x$ and $v R_{j} y$, so $\mathcal{M}, x \models p$ and $\mathcal{M}, y \models q$, and, thus, $\mathcal{M}, w' \models p \bullet q$ and $\mathcal{M}, w \models \: !_{k} (p \bullet q)$.

\item $\mathcal{M}, w \models {\bf 1} \vdash \: !_i {\bf 1}$ follows from the condition $[R_i]\mathcal{O} = \mathcal{O}$.
\item Assume that $\mathcal{M}, w \models \: !_{i} p \: \bullet \: q$ and $i \in E$, then, by Proposition~\ref{val}, $w \in [R_i][\![p]\!] \cdot [\![q]\!]$. But $i \in E$, then $w \in [\![q]\!] \cdot [R_i][\![p]\!]$, then  $\mathcal{M}, w \models \: q \: \bullet \: !_{i} p$.
\item Assume that $\mathcal{M}, w \models \: !_{i} p \: \bullet \: q$ and $i \in C$, then, by Proposition~\ref{val}, $w \in [R_i][\![p]\!] \cdot [\![q]\!]$, so, as far as $i \in C$, $w \in [R_i][\![p]\!] \cdot [\![q]\!] \cdot [R_i][\![p]\!]$, then $\mathcal{M}, w \models \: !_{i} p \: \bullet \: q \: \bullet \: !_{i} p$. Another non-local contraction case can be proved similarly.
\item $\mathcal{M}, w \models \: !_i p \: \vdash \: {\bf 1}$ for $i \in W$ follows from the condition $[R_i][\![p]\!] \subseteq \mathcal{O}$.
\end{enumerate}
Notice that validity of axioms $!_i p \: \vdash \: p$ and $!_i p \: \vdash \: !_i !_i p$ follows from the condition that each $R_i$ is reflexive and transitive,
the proof is completely identical to \cite[Proposition 3.30 and Proposition 3.31]{DBLP:books/daglib/0030819}.
\end{proof}

Also, we connect validity in $\Sigma$-frames and their complex algebras as follows:
\begin{proposition} \label{complex}
Let $\mathcal{F}$ be a $\Sigma$-frame, then $\mathcal{F} \models \varphi \vdash \psi$ iff $\operatorname{Cm} \models [\![\varphi]\!] \leq [\![\psi]\!]$.
\end{proposition}
\begin{proof}
Follows from Proposition~\ref{val}.
\end{proof}

\section{Canonical extensions}

In this section, we introduce canonical extensions of $\Sigma$-algebras for an arbitrary subexponential signature $\Sigma$ and then use them to show Kripke completeness of $\operatorname{DSMALC}_{\Sigma}$ with respect to its frames.
Alternatively, we can show that $\operatorname{DSMALC}_{\Sigma}$ is complete with respect to Kripke semantics
by constructing the canonical model of prime theories, but we preferred this more algebraic approach to show both Kripke completeness and canonicity of the corresponding variety of algebras.

Our approach is based on canonical extensions of distributive bounded lattice expansions in the fashion of \cite{gehrke2004bounded} and \cite{gehrke2005sahlqvist}. We recall some underlying notions first of all.

Let $\mathcal{L}$ be a (complete) lattice and $a \in \mathcal{L}$, then $a$ is \emph{completely join-irreducible} if $a = \bigvee \limits_{i \in I} a_i$ implies that there is $j \in I$ such that $a = a_j$. Completely meet-irreducibeles are defined dually. $\mathcal{J}^{\infty}(\mathcal{L})$ ($\mathcal{M}^{\infty}(\mathcal{L})$) is the set of all completely join-irreducible (meet-irreducible) elements.

Note that $\mathcal{J}^{\infty}(\mathcal{L})$ and $\mathcal{M}^{\infty}(\mathcal{L})$ are order-isomorphic, the isomorphism $\kappa : \mathcal{J}^{\infty}(\mathcal{L}) \to \mathcal{M}^{\infty}(\mathcal{L})$ is defined as $\kappa : j \mapsto \bigvee (- \uparrow j)$, see, e.g., \cite[Theorem 2.3]{gehrke2004bounded}.

Let us define canonical extensions for bounded distributive lattices first and then extend it for required expansions. Recall that a complete distributive lattice $\mathcal{L}$ is called \emph{perfect} if:
\begin{enumerate}
\item $\mathcal{L}$ is completely distributive, that is, for every doubly indexed family $(a_{i,j})_{i \in I, j \in J}$ of $\mathcal{L}$, one has:
\begin{center}
$\bigwedge \limits_{i \in I} \bigvee \limits_{j \in J} a_{i,j} = \bigvee \limits_{f : J \to I} \bigwedge \limits_{i \in I} a_{i, f(i)}$
\end{center}
\item and every $a \in \mathcal{L}$ can be expressed as
\begin{center}
$a = \bigvee \{ j \in \mathcal{J}^{\infty}(\mathcal{L}) \: | \: j \leq x \} = \bigwedge \{ m \in \mathcal{M}^{\infty}(\mathcal{L}) \: | \: x \leq m \}$.
\end{center}
\end{enumerate}
 A \emph{canonical extension} of a bounded distributive lattice $\mathcal{L}$ is a perfect lattice $\mathcal{L}^{\sigma}$ such that $\mathcal{L} \hookrightarrow \mathcal{L}^{\sigma}$ and such that the compactness property holds:
\begin{itemize}
\item If $S, T \subseteq \mathcal{L}$ such that $\bigwedge S \leq \bigvee T$ in $\mathcal{L}^{\sigma}$, then there are finite subsets $S' \subseteq S$ and $T' \subseteq T$ such that $\bigwedge S' \leq \bigvee T'$ in $\mathcal{L}$.
\end{itemize}

It is known that every bounded distributive lattice has a unique canonical extension, which is, in fact, based on the Stone-Priestley representation \cite{priestley1970representation}.

To define canonical extensions for $\Sigma$-algebras, we combine techiques for residuated lattices and distributive modal algebras developed in \cite{gehrke2021topological} and \cite{gehrke2005sahlqvist}. Let $\mathcal{R} = (R, \vee, \wedge, \cdot, \setminus, /, (!_i)_{i \in \Sigma}, 0, 1, \epsilon)$ be a $\Sigma$-algebra, then $\mathcal{A}$ is \emph{perfect} if the following holds:
\begin{enumerate}
\item The lattice reduct is perfect as a distributive lattice,
\item $!$ preserves all infima,
\item $\cdot$ is completely additive, that is, it preserve all suprema in both arguments,
\item The residuals $\setminus : \mathcal{R} \times \mathcal{R} \to \mathcal{R}$ and $/ : \mathcal{R} \times \mathcal{R} \to \mathcal{R}$ are complete in the following sense: $\setminus$ preserves all infima in the first argument and all suprema in the second one and $/$ preserves all suprema in the first agrument and all infima in the second one.
\end{enumerate}

Given a $\Sigma$-algebra $\mathcal{R} = (R, \wedge, \vee, \cdot, \setminus, /, (!_i)_{i \in \Sigma}, 0, 1, \epsilon)$, the canonical extension of its lattice reduct $\mathcal{R}^{\sigma}$ can be extended to the canonical extension of a $\Sigma$-algebra. First of all, we define the sets of \emph{filter} and \emph{ideal} elements of the canonical extension $\mathcal{R}^{\sigma}$ as
\begin{center}
$\mathcal{F}(\mathcal{R}^{\sigma}) = \{ x \in \mathcal{R}^{\sigma} \: | \: \text{$x$ is a meet of elements from $\mathcal{R}$}\}$,

$\mathcal{I}(\mathcal{R}^{\sigma}) = \{ x \in \mathcal{R}^{\sigma} \: | \: \text{$x$ is a join of elements from $\mathcal{R}$} \}$.
\end{center}

First of all, we define residuals, modalities and the product operation for filter and ideal elements and then extend them for arbitrary elements of the canonical extension $\mathcal{R}^{\sigma}$ of the lattice reduct of $\mathcal{R}$. Given $x, x' \in \mathcal{F}(\mathcal{R}^{\sigma})$ and $y,y' \in \mathcal{I}(\mathcal{R}^{\sigma})$, then:

\begin{itemize}
\item $x \cdot^{\sigma} x' = \bigwedge \{ a \cdot a' \: | \: x \leq a \in \mathcal{R} \: \& \: x' \leq a' \in \mathcal{R} \}$,
\item $x \setminus^{\pi} y = \bigvee \{ a \setminus b \: | \: x \leq a \in \mathcal{R} \ni b \leq y \}$,
\item $!_i^{\sigma} x = \bigwedge \{ !_i a \: | \: x \leq a \in \mathcal{R} \}$, $i \in \Sigma$.
\end{itemize}

So we define the canonical extesion $\mathcal{R}^{\sigma}$ of $\mathcal{R}$ as the algebra $\mathcal{R}^{\sigma} = (R^{\sigma}, \cdot^{\sigma}, \setminus^{\pi}, /^{\pi},  (!_i^{\sigma})_{i \in \Sigma}, \epsilon)$ such that for all $a, b \in \mathcal{R}^{\sigma}$:
\begin{itemize}
\item $a \cdot^{\sigma} b = \bigvee \{ a' \cdot^{\sigma} b' \: | \: \mathcal{F}(\mathcal{R}^{\sigma}) \ni a' \leq a \: \& \: \mathcal{F}(\mathcal{R}^{\sigma}) \ni b' \leq b \}$,
\item $a \setminus^{\pi} b = \bigwedge \{ a' \setminus^{\pi} b' \: | \: a \geq a' \in \mathcal{I}(\mathcal{R}^{\sigma}) \: \& \: b \geq b' \in \mathcal{F}(\mathcal{R}^{\sigma})\}$,
\item $!_i^{\sigma} a = \bigvee \{ !^{\sigma} a' \: | \: \mathcal{F}(\mathcal{R}^{\sigma}) \ni a' \leq a \}$, $i \in \Sigma$.
\end{itemize}

The definition of $/^{\pi}$ is right-to-left symmetric to $\setminus^{\pi}$.

\begin{theorem} \label{canonical}
The variety of all $\Sigma$-algebras is canonical, that is, it is closed under canonical extensions.
\end{theorem}

\begin{proof}
The fact that the canonical extension of the lattice reduct is a perfect distributive lattice is by Jonsson, see \cite[Theorem 2.3]{gehrke1994bounded}. The canonical extension of the residuated lattice reduct of $\mathcal{R}$ is also a perfect residuated lattice, see \cite[Proposition 5]{gehrke2021topological}. The canonical extension of the $!$ modal operator preserves arbitrary infima as it was shown in \cite[Lemma 2.21]{gehrke2005sahlqvist}. The items we have got to check are the following, for $i, j, k \in \Sigma$ and for $a, b \in \mathcal{R}^{\sigma}$:
\begin{enumerate}
\item $!^{\sigma}_i a \: \cdot^{\sigma} \: !^{\sigma}_j b \leq \: !_k^{\sigma}(a \cdot^{\sigma} b)$ whenever $k \preceq i, j$,
\item $\epsilon \leq !^{\sigma}_i \epsilon$
\item $!^{\sigma}_i a \: \cdot^{\sigma} b = b \: \cdot^{\sigma}  \: !^{\sigma}_i a$ whenever $i \in E$,
\item $!^{\sigma}_i a \: \cdot^{\sigma} b \leq \: !^{\sigma}_i a \cdot^{\sigma} b \: \cdot^{\sigma} \: !^{\sigma}_i a$ and $b \: \cdot^{\sigma} \: !^{\sigma}_i a \leq \: !^{\sigma}_i a \cdot^{\sigma} b \: \cdot^{\sigma} \: !^{\sigma}_i a$ whenever $i \in C$,
\item $!^{\sigma}_i a \: \leq \epsilon$ whenever $i \in W$.
\end{enumerate}
We check only the first and fourth items, the rest of them are shown similarly.
\begin{enumerate}
\item Fix $i,j, k \in \Sigma$ such that $k \preceq i, j$. First of all, take any $x, y \in \mathcal{F}(\mathcal{R}^{\sigma})$. Observe that
\begin{center}
$!^{\sigma}_i x \: \cdot^{\sigma} \: !^{\sigma}_j y = \bigwedge \{ !_i \: x' \cdot \: !_j y' \: | \: x \leq x' \in \mathcal{R} \: \& \: y \leq y' \in \mathcal{R}\}$
\end{center}
which follows from that fact $\cdot^{\sigma}$ is order-preserving and both $!^{\sigma}_i$ and $!^{\sigma}_j$ preserve arbitrary infima and from the definition of a filter element. Then we have:

$\begin{array}{lll}
!^{\sigma}_i x \: \cdot^{\sigma} \: !^{\sigma}_j y = && \\
\:\:\:\:\:\:\:\:\:\:\:\: \text{By the definition of $\cdot^{\sigma}$} \\
\:\:\:\:\:\:\:\: \bigwedge \{ !_i \: x' \cdot \: !_j y' \: | \: x \leq x' \in \mathcal{R} \: \& \: y \leq y' \in \mathcal{R}\} \leq && \\
\:\:\:\:\:\:\:\:\:\:\:\:  \text{By the axiom of $\Sigma$-algebras} \\
\:\:\:\:\:\:\:\: \bigwedge \{ !_k (x' \cdot \: y') \: | \: x \leq x' \in \mathcal{R} \: \& \: y \leq y' \in \mathcal{R}\} = && \\
\:\:\:\:\:\:\:\:\:\:\:\:  \text{By the definition of $!_k^{\sigma}$} \\
\:\:\:\:\:\:\:\: \bigwedge !_k^{\sigma} \{ x' \cdot \: y' \: | \: x \leq x' \in \mathcal{R} \: \& \: y \leq y' \in \mathcal{R}\} = && \\
\:\:\:\:\:\:\:\:\:\:\:\: \text{$!_k^{\sigma}$ commutes with all infima} \\
\:\:\:\:\:\:\:\: !_k^{\sigma} \bigwedge \{ x' \cdot \: y' \: | \: x \leq x' \in \mathcal{R} \: \& \: y \leq y' \in \mathcal{R}\} = && \\
\:\:\:\:\:\:\:\:\:\:\:\: \text{By the definition of $\cdot^{\sigma}$} \\
\:\:\:\:\:\:\:\: !_k^{\sigma} (x \: \cdot^{\sigma} y)
\end{array}$

Now we show it for arbitrary elements of $\mathcal{R}^{\sigma}$. So take any $a, b \in \mathcal{R}^{\sigma}$, then, using monotonicity of $!_k^{\sigma}$ and the observation above, we have:

$\begin{array}{lll}
!^{\sigma}_i a \: \cdot^{\sigma} \: !^{\sigma}_j b = && \\
\:\:\:\:\:\:\:\:\:\:\:\: \text{By the definition of $!^{\sigma}_i$, $!^{\sigma}_j b$ and $\cdot^{\sigma}$ } && \\
\:\:\:\:\:\:\:\: \bigvee \{ !^{\sigma}_i x \: \cdot^{\sigma} \: !^{\sigma}_j y  \: | \: a \geq x \in
 \mathcal{F}(\mathcal{R}^{\sigma}) \: \& \: b \geq y \in \mathcal{F}(\mathcal{R}^{\sigma}) \} \leq && \\
\:\:\:\:\:\:\:\:\:\:\:\: \text{The observation above} && \\
\:\:\:\:\:\:\:\: \bigvee \{ !^{\sigma}_k (x \cdot^{\sigma} y)  \: | \: a \geq x \in \mathcal{F}(\mathcal{R}^{\sigma}) \: \& \: b \geq y \in \mathcal{F}(\mathcal{R}^{\sigma}) \} \leq && \\
\:\:\:\:\:\:\:\:\:\:\:\: \text{$!^{\sigma}_k$ is order-preserving} &&\\
\:\:\:\:\:\:\:\: !^{\sigma}_k \bigvee \{ x \cdot^{\sigma} y  \: | \: a \geq x \in \mathcal{F}(\mathcal{R}^{\sigma}) \: \& \: b \geq y \in \mathcal{F}(\mathcal{R}^{\sigma}) \} = && \\
\:\:\:\:\:\:\:\:\:\:\:\: \text{The definition of $\cdot^{\sigma}$} && \\
\:\:\:\:\:\:\:\: !^{\sigma}_k(a \cdot^{\sigma} b)
\end{array}$
\item Now let us show that $!^{\sigma}_i a \: \cdot^{\sigma} b \leq \: !^{\sigma}_i a \cdot^{\sigma} b \: \cdot^{\sigma} \: !^{\sigma}_i a$ for $i \in C$ and $a, b \in \mathcal{R}^{\sigma}$. Let us show this inequality for filter elements first. Let $x, y \in \mathcal{F}(\mathcal{R}^{\sigma})$, then

$\begin{array}{lll}
!^{\sigma}_i x \: \cdot^{\sigma} y = && \\
\:\:\:\:\:\:\:\:\:\:\:\: \text{The definition of filter elements and $!^{\sigma}_i$} && \\
\:\:\:\:\:\:\:\: \bigwedge \{ !_i x' \: | \: x \leq x' \in \mathcal{R} \} \: \cdot^{\sigma} \bigwedge \{ y' \in \mathcal{R} \: | \: y \leq y' \} = && \\
\:\:\:\:\:\:\:\:\:\:\:\: \text{The definition of $\cdot^{\sigma}$} && \\
\:\:\:\:\:\:\:\: \bigwedge \{ !_i x' \: \cdot \: y' \: | \: x \leq x' \in \mathcal{R}, y \leq y' \in \mathcal{R} \} \leq && \\
\:\:\:\:\:\:\:\:\:\:\:\: \text{The axioms of $\Sigma$-algebras} && \\
\:\:\:\:\:\:\:\: \bigwedge \{ !_i x' \: \cdot \: y' \: \cdot \: !_i x' \: | \: x \leq x' \in \mathcal{R}, y \leq y' \in \mathcal{R} \} = && \\
\:\:\:\:\:\:\:\:\:\:\:\: \text{The definition of $\cdot^{\sigma}$} && \\
\:\:\:\:\:\:\:\: \bigwedge \{ !_i x' \: | \: x \leq x' \in \mathcal{R} \} \: \cdot^{\sigma} \bigwedge \{ y' \in \mathcal{R} \: | \: y \leq y' \} \cdot^{\sigma} \bigwedge \{ !_i x' \: | \: x \leq x' \in \mathcal{R} \} = && \\
\:\:\:\:\:\:\:\:\:\:\:\: \text{$!^{\sigma}_i$ preserves all infima} && \\
\:\:\:\:\:\:\:\: !^{\sigma}_i \bigwedge \{ x \: | \: x \leq x' \in \mathcal{R} \} \: \cdot^{\sigma} \bigwedge \{ y \in \mathcal{R} \: | \: y \leq y' \} \: \cdot^{\sigma} !^{\sigma}_i \bigwedge \{ x \: | \: x \leq x' \in \mathcal{R} \} = && \\
\:\:\:\:\:\:\:\:\:\:\:\: \text{The definition of $\cdot^{\sigma}$} && \\
\:\:\:\:\:\:\:\: !^{\sigma}_i x \: \cdot^{\sigma} y \: \cdot^{\sigma} \: !^{\sigma}_i x
\end{array}$

So, for arbitrary $a, b \in \mathcal{R}^{\sigma}$ we have:

$\begin{array}{lll}
{!}^{\sigma}_i a \: \cdot^{\sigma} b = &&\\
\:\:\:\:\:\:\:\:\:\:\:\: \text{The definition of $\cdot^{\sigma}$} && \\
\:\:\:\:\:\:\:\: \bigvee \{ !^{\sigma}_i x \: \cdot^{\sigma} y \: | \: a \geq x \in \mathcal{F}(\mathcal{R}^{\sigma}), y \geq \in \mathcal{F}(\mathcal{R}^{\sigma}) \} \leq && \\
\:\:\:\:\:\:\:\:\:\:\:\: \text{The observation above} && \\
\:\:\:\:\:\:\:\: \bigvee \{ !^{\sigma}_i x \: \cdot^{\sigma} y \: \cdot^{\sigma} !^{\sigma}_i x \: | \: a \geq x \in \mathcal{F}(\mathcal{R}^{\sigma}), y \geq \in \mathcal{F}(\mathcal{R}^{\sigma}) \} = && \\
\:\:\:\:\:\:\:\:\:\:\:\: \text{The definition of $\cdot^{\sigma}$ and $!^{\sigma}_i$} && \\
\:\:\:\:\:\:\:\: !^{\sigma}_i a \cdot^{\sigma} b \: \cdot^{\sigma} \: !^{\sigma}_i a
\end{array}$
\end{enumerate}
\end{proof}
In the previous section, we defined complex algebras of $\Sigma$-frames, so let us define dual structures of perfect $\Sigma$-algebras. Let $\Sigma$ be a subexponential signature and let $\mathcal{A} = (A, \bigvee, \bigwedge, \cdot, \setminus, /, (!_i)_{i \in I}, 0, 1, \epsilon)$ be a $\Sigma$-algebra, its \emph{atom structure} (that is, the frame of completely join-irreducible elements) is the structure of the form $\operatorname{At}(\mathcal{A}) = (\mathcal{J}^{\infty}(\mathcal{A}), \leq_{\delta}, R, (R_i)_{i \in \Sigma}, \uparrow \epsilon)$, where $\leq_{\delta}$ is dual order on $\mathcal{J}^{\infty}(\mathcal{R})$ and:
\begin{itemize}
\item $R$ is a ternary relation on $\mathcal{J}^{\infty}(\mathcal{A})$ such that $Rabc$ iff $a \cdot b \leq_{\delta} c$,
\item $R_i$ is a binary relation on $\mathcal{J}^{\infty}(\mathcal{A})$ such that $a R_i b$ iff $!_i \kappa(a) \leq_{\delta} \kappa(b)$ for each $i \in \Sigma$.
\end{itemize}
The definition of perfect $\Sigma$-algebras guarantee that $\operatorname{At}(\mathcal{A})$ is a well-formed Kripke-frame. $\operatorname{At}$ and $\operatorname{Cm}$ are connected with each other as follows:
\begin{lemma} \label{duality}
Let $\mathcal{R}$ be a perfect $\Sigma$-algebra, then $\mathcal{R} \cong \operatorname{Cm}(\operatorname{At}(\mathcal{R}))$.
\end{lemma}
\begin{proof}
We give only a proof for products, the modal part is similar to \cite[Proposition 2.25]{gehrke2005sahlqvist}, the argument for products residuals is similar to \cite[Lemma 6.10]{galatos2003varieties}. The isomorphism itself is established with the map $\eta : a \mapsto \{ b \in \mathcal{J}^{\infty}(\mathcal{R}) \: | \: a \leq_{\delta} b \}$. That is, we extend Raney representation of perfect distibutive lattices \cite{raney1952completely}.
\end{proof}

Finally, using canonical extensions for $\Sigma$-algebras and duality between perfect $\Sigma$-algebras and $\Sigma$-frames.

\begin{theorem}
$\operatorname{DSMALC}_{\Sigma}$ is complete with respect to Kripke semantics.
\end{theorem}
\begin{proof}
Suppose that $\varphi \nvdash \psi$, then $F_{\Sigma} \not\models \varphi \leq \psi$ by algebraic completeness, then $F_{\Sigma}^{\sigma} \not\models \varphi \leq \psi$ since $F_{\Sigma}^{\sigma}$ contains $F_{\Sigma}$ as a subalgebra,  $F_{\Sigma}^{\sigma}$ is a perfect $\Sigma$-algebra as it was shown in Theorem~\ref{canonical}. By Lemma~\ref{duality}, $\operatorname{Cm}(\operatorname{At}(F_{\Sigma}^{\sigma})) \not\models \varphi \leq \psi$ and therefore $\operatorname{At}(F_{\Sigma}^{\sigma}) \not\models \varphi \vdash \psi$ by Proposition~\ref{complex}.
\end{proof}

\begin{remark}
$\operatorname{At}(F_{\Sigma}^{\sigma})$, the atom structure of the canonical extension of the Lindenbaum-Tarksi algebra in the variety of all $\Sigma$-alebras can be thought as the \emph{canonical frame} of $\operatorname{DSMALC}_{\Sigma}$.
\end{remark}

\bibliographystyle{eptcs}
\bibliography{generic}
\end{document}